 \documentclass{ptptex}

\usepackage{amsthm}
\newtheorem{theorem}{Theorem} 

\newcommand{\xG}{\mathbb{G}}
\newcommand{\xe}{e}			
\newcommand{\IsDef}{\triangleq}			

\title{
	First-order transition in Potts models with ``invisible'' states
}

\subtitle{Rigorous proofs}    

\author{
	Aernout~C.~D. \textsc{van Enter},%
		\footnote{E-mail: \texttt{A.C.D.van.Enter@.rug.nl}} \ 
	Giulio \textsc{Iacobelli},%
		\footnote{E-mail: \texttt{dajegiu@gmail.com}} \
	Siamak \textsc{Taati}
}

\inst{
	Johann Bernoulli Institute,
	Rijksuniversiteit Groningen,
	Nijenborgh 9, 9747~AG~Groningen,
	The Netherlands
}

\abst{
	In some recent papers by Tamura, Tanaka and Kawashima,\cite{TTK1, TTK2,TT} \
	a class of Potts models with ``invisible'' states was introduced, for which 
	the authors argued by numerical arguments and by a mean-field analysis that a 
	first-order transition occurs. Here we show that the existence of this 
	first-order transition can be proven rigorously, by relatively minor adaptations 
	of existing proofs for ordinary Potts models. In our argument we  present a 
	random-cluster representation for the model, which might be of independent 
	interest.
}

\begin{document}

\maketitle

\section{Introduction}
In Refs.~\citen{TTK1,TTK2}, the authors introduced a class of Potts models, in which 
next to $q$ ordinary -visible- colours (the Potts states), between  which  a 
standard ferromagnetic nearest-neighbour Potts interaction exists, $r$
``invisible'' colours (states) are possible, which have zero interaction 
energy with any neighbour, whatever state that neighbour is in.  

Although the number of ground states and low-temperature states equals $q$, and 
there is at low temperatures 
spontaneous symmetry breaking of the $q$-fold permutation symmetry
just as in the standard $q$-state Potts model,  
the transition for low $q=2,3,4$ and high $r$ is different from the second-order
transition of the ordinary two-dimensional $q$-state Potts model. In 
fact a first-order transition in the temperature-parameter appears. 

The occurrence of such a first-order transition contradicts a simple form of 
universality which would predict that all systems with the same broken symmetry 
in the same dimension with short-range interactions 
have the same type of transition.

However, such a universality property is known to be too strong to be true. 
The question of first-order versus second-order is {\em not} a universal question. 
Some counterexamples illustrating this point are the two-dimensional $3$-state 
Kac-Potts model,\cite{GobMer} \ in which a 
first-order transition occurs in presence of a broken $3$-fold rotation 
(= permutation) symmetry, or the three-dimensional versions of the nonlinear 
$O(n)$-models treated in Ref.~\citen{DSS,BGH,ES1,ES2,ES3}, in which a first-order 
transition occurs in presence of a broken continuous rotation symmetry. 
In both cases the same type of symmetry-breaking is also known to be possible 
with a second-order transition. This  occurs for the standard 
nearest-neighbour $3$-state two-dimensional Potts model, or for the standard 
three-dimensional classical Heisenberg or XY models respectively.   

In fact, the model with many invisible states has a first-order transition 
for the same reason the high-$q$ Potts model has a first-order transition. 
At the transition temperature there is coexistence between  
high-energy phases and a high-entropy phase, and between them 
``free-energy barriers'' exist. Such a coexistence can be proven by a 
form of a Peierls-type free-energy-contour argument. For the standard 
Potts model, by now there exists a variety of
such proofs,\cite{KS, BKL,LMMRS,M} \
whether by Reflection Positivity and Chessboard Estimates, or by a 
Pirogov-Sinai argument, either within a spin 
description or in a random-cluster version. Typically those proofs  can be 
adapted without too much effort to include the model described above. 

Inside  standard Potts ``free-energy contours'', as were 
described above,  sites 
exist on the border between  the ordered and disordered phases and 
hence they are neither ordered nor disordered themselves. They have neither all 
neighbours different, nor all neighbours equal, and as a consequence, in dimension
$d$  they lose a free-energy fraction per site of order 
$\frac{1}{2d}$ with respect to the free energy of either a disordered site 
(where free energy  is purely entropic) or with respect to the free energy of 
an ordered site (whose free energy is purely energetic). At the transition 
temperature the entropy of a disordered site approximately 
equals the energy of an ordered site. 

In  the model we will consider below, 
the ordered sites need to be in a visible colour, whereas 
disordered sites can have all neighbours either different or invisible.   
It is enough to  consider only the two-dimensional version, but this is 
not essential, and the arguments directly generalize to higher dimensions.
As the presence of a first-order transition in higher-dimensional Potts 
models is less surprising, the main interest seems to be in two dimensions.  

As  a further comment we mention that the term ``invisible'' is actually 
a bit of a misnomer, as at high temperatures  
the density of ``invisible'' colours is higher than those of the 
``visible'' ones when $r \geq q$. Thus most of the colours which appear would be
the ``invisible'' ones.

\section{Main result}
At each site  there is  a discrete-valued spin variable which 
can take one out of $q+r$ colours, $q$ of which ``visible'' and 
$r$ ``invisible''. 
The $(q,r)$-model then is defined by the following (formal) Hamiltonian
\begin{equation*}
	H=-J\sum_{\langle i,j\rangle}\delta(\sigma_i, \sigma_j)
	\sum_{\alpha=1}^q\delta(\sigma_i, \alpha)
\end{equation*} 
Here the pairs of nearest-neighbour sites $\langle i,j\rangle$
live on a lattice of dimension at least two.
Now we have the following result.

\medskip

\begin{theorem}
	For $q+r$ large enough the above model undergoes a first-order transition 
	in temperature. At the transition temperature $q$ ordered extremal 
	Gibbs states coexist with a disordered extremal Gibbs state.
\end{theorem}

\begin{proof}[Proofs]
There are various ways in which one may adapt existing proofs. For example, 
the proof originally due to Koteck\'y and Shlosman,\cite{KS} \ 
later also treated in Refs.~\citen{Geo,Shl}, 
could be adapted by observing that 
\begin{itemize}
	\item Our model has a $C$-potential (in Georgii's terms),
		so reflection positivity holds.
	\item An ordered bond now will be a bond whose two sites have the same
		{\em visible} colour.
	\item The ``restricted ensemble'' for the disordered phase
		is formed by all configurations having disordered bonds only,
		which has an approximate entropy density $\ln (q+r)$,
		when $q+r$ is large.
\end{itemize}
Then the arguments used in section 19.3 of Ref.~\citen{Geo}
or in Ref.~\citen{Shl}, 
using chessboard estimates to provide a contour estimate, apply.

Here we will sketch  in some more detail an alternative
proof based on the Fortuin-Kasteleyn 
random-cluster  description, first derived in Ref.~\citen{LMMRS}, 
and later treated e.g. in Ref.~\citen{Grim}. 
This has the advantage that it extends to values of $q$ and $r$
which need not  correspond to a spin model interpretation, e.g. $q=1$,
or $q$ and{/}or $r$ non-integer.
To do this we will adapt the ordinary random-cluster representation to 
include the model with invisible states.

In analogy with the standard Potts model,\cite{ForKas72,KF} \ 
it is possible to rewrite the partition function for 
the $(q,r)$-Potts model in terms of the partition function
for a variant of the random-cluster model, which we will call the 
``$r$-biased'' random cluster model.
Just as the standard random-cluster model, the $r$-biased model is a 
correlated bond-percolation model.

Let $\xG=(S,B)$ be a finite graph, where $S$ denotes the set of 
sites, and $B$ the set of bonds in the graph. 
The \emph{$r$-biased random-cluster} model on $\xG$ is given by a
probability distribution on the sets $X\subseteq B$.
The distribution 
has three parameters $0\leq p\leq 1$, $q>0$ and $r\geq 0$
and is defined by
\begin{align*}
	\phi_{p,q,r}(X) &=
		\frac{1}{Z^{RC}_{p,q,r}(\xG)}
		\left[\prod_{b\in B}
			p^{\delta(b\in X)}(1-p)^{\delta(b\notin X)}\right]
			(q+r)^{\kappa_0(S,X)}q^{\kappa_1(S,X)} \;,
\end{align*}
in which $\kappa_0(S,X)$ denotes the number of
isolated vertices of the graph $(S,X)$,
$\kappa_1(S,X)$ the number of
non-singleton connected components of $(S,X)$ and
$Z^{RC}_{p,q,r}(\xG)$ the partition function.
Notice that for $r=0$, the model reduces to the standard
random-cluster model, in which both singleton and non-singleton
connected components have weight $q$.
For $r>0$, the above model induces a bias
towards singleton connected components.
Namely, the singleton connected components have weight $(q+r)$
whereas the non-singleton connected components have weight $q$.

Let us now see how the $(q,r)$-Potts model is related
to the $r$-biased random-cluster model.
Let $\Omega$ be the set of $(q,r)$-Potts configurations on $\xG$.
The partition function of this model can be rewritten as
\begin{subequations}
\begin{align*}
	Z_\beta(\xG) &=
		\sum_{\sigma\in\Omega}\xe^{
			\beta\sum_{\{i,j\}\in B} \delta(\sigma_i=\sigma_j\leq q)
		} \\
	&=		\sum_{\sigma\in\Omega} \prod_{\{i,j\}\in B}
			\xe^{\beta\delta(\sigma_i=\sigma_j\leq q)} \\
	&=		\sum_{\sigma\in\Omega} \prod_{\{i,j\}\in B}
			\left[1+\delta(\sigma_i=\sigma_j\leq q)(\xe^\beta-1)\right] \\
	&=	\sum_{\sigma\in\Omega}\sum_{X\subseteq B}
			\prod_{\{i,j\}\in X}\delta(\sigma_i=\sigma_j\leq q)
			(\xe^\beta-1)^{|X|} \\
	&=		\sum_{\sigma\in\Omega}\sum_{X\subseteq B}
			\pi(\sigma,X) \;,
\end{align*}
\end{subequations}
where
\begin{align*}
	\pi(\sigma,X) &=
		\xe^{\beta|B|}
		\prod_{\{i,j\}\in B}
		\left[
			\delta(\{i,j\}\in X)\delta(\sigma_i=\sigma_j\leq q)(1-\xe^{-\beta}) +
			\delta(\{i,j\}\notin X)\xe^{-\beta}	
		\right]\;.
\end{align*}
The  expression above  describes a coupling of the
$(q,r)$-Potts distribution on $\Omega=\{1,...q+r \}^{S}$
and a probability distribution on the space $\{0, 1\}^{B}$
(compare Ref.~\citen{ES}).
The marginal of this coupling on the space $\{0,1\}^{B}$
is simply the $r$-biased random-cluster distribution
$\phi_{p_\beta,q,r}$ with $p_\beta=1-\xe^{-\beta}$.
In particular, the weight $\pi(\sigma,X)$ can also be expressed as
\begin{align*}
	\pi(\sigma,X) &=
		\xe^{\beta|B|}\cdot 1_{F_r}(\sigma,X)\cdot
		\prod_{\{i,j\}\in B}
                \left[
			p_\beta\,\delta(\{i,j\}\in X) +		
			(1-p_\beta)\,\delta(\{i,j\}\notin X)
		\right] \;,
\end{align*}
where
\begin{align*}
	F_r &\IsDef\left\{
		(\sigma,X): \sigma_i=\sigma_j\leq q
		\text{ for all } \{i,j\}\in X
	\right\} \;.
\end{align*}

This expression gives us the model with free boundary conditions.
The wired boundary conditions in the $(q,r)$-model correspond to all spins on 
the boundary having the same visible colour. 

The Peierls (free-energy) contour estimate is almost unchanged in comparison 
with the standard case. The only difference is that,
(except for the degeneracy term of the free contours) the
contours of Ref.~\citen{LMMRS} satisfy the same
Peierls estimate with~$q$ replaced by~$q+r$.
More precisely, in Ref.~\citen{Grim}, in equation~(7.59), the adapted Peierls 
estimate for wired contours is of the form 
\begin{equation*}
	\Phi_w(\gamma_w) \leq (q+r)^{-\frac{\Vert\gamma_w\Vert}{2d}}
	\;\xe^{ 5\Vert\gamma_w\Vert}\;,
\end{equation*}
whereas the corresponding estimate for free contours, below
equation~(7.61) of Ref.~\citen{Grim},
takes the form 
\begin{equation*}
	\Phi_f(\gamma_f) \leq q\;(q+r)^{-\frac{\Vert\gamma_f\Vert}{2d}}
	\;\xe^{5 \Vert\gamma_f\Vert}\;. 
\end{equation*}
Convergence of the Peierls estimate will hold once~$q+r$ is large enough, as 
we then have (both for free and wired contours) an estimate of the form
\begin{equation*}
	\Phi(\gamma) \leq \xe^{-\tau\Vert\gamma\Vert}\;,
\end{equation*}
for a sufficiently large~$\tau$.

For a more detailed treatment we refer to Refs.~\citen{eit2,I}.
\end{proof}

\section{Comments and conclusions}
In this note we have shown how the Potts model with many invisible states, 
introduced by Tamura, Tanaka and Kawashima, can be proven to have a first-order 
phase transition, similarly as occurs for the standard high-$q$ Potts models.
The transition temperature is asymptotically given by $\beta\approx\ln(q+r)$. 
The proofs, as usual, apply for quite high values of $q$ and/or $r$, the 
numerical approach of Refs.~\citen{TTK1,TTK2} might give a better indication of the values at which the first-order transition first occurs.  

We conjecture that the  dynamical properties of the Potts model with $r$ 
invisible states, which were considered in Ref.~\citen{TT},
have a similar corresponding behaviour as occurs in the ordinary Potts model,
for which they were rigorously  analysed in Ref.~\citen{BCKFTVV}.

\section*{Acknowledgements}
G.~I. and S.~T. thank NWO for support. 

%

\end{document}